\pdfoutput=1
\documentclass{amsart}
\usepackage[utf8]{inputenc}
\usepackage[english]{babel}

\usepackage{mathptmx}

\usepackage{amsmath}
\usepackage{amstext}
\usepackage{amssymb}

\usepackage{microtype}
\usepackage{lmodern}

\usepackage{fancyhdr}

\usepackage[retainorgcmds]{IEEEtrantools}
\usepackage{amsfonts}

\usepackage{tikz}
\usetikzlibrary{matrix,arrows,calc}

\usepackage[]{algpseudocode}

\usepackage{enumerate}

\usepackage{hyperref}

\allowdisplaybreaks[4]

\theoremstyle{definition}
\newtheorem{theorem}{Theorem}[section]
\newtheorem{definition}[theorem]{Definition}
\newtheorem{remark}[theorem]{Remark}
\newtheorem{lemma}[theorem]{Lemma}

\newtheorem{algorithm}[theorem]{Algorithm}

\newenvironment{alg}{
\begin{algorithmic}[1]
\normalfont
}{
\end{algorithmic}
}

\newcommand{\mb}[1]{\mathbb{#1}}

\newcommand{\N}{\mb N}
\newcommand{\Z}{\mb Z}
\newcommand{\R}{\mb R}
\newcommand{\Q}{\mb Q}
\renewcommand{\epsilon}{\varepsilon}
\newcommand{\F}{\mathcal{F}}

\DeclareMathOperator{\Aff}{Aff}

\DeclareMathOperator{\GL}{GL}

\DeclareMathOperator{\Orb}{Orb}

\DeclareMathOperator{\Span}{Span}
\DeclareMathOperator{\diam}{diam}

\newcommand{\m}[1]{ \begin{pmatrix} #1 \end{pmatrix} }

\usepackage[foot]{amsaddr}

\title[Canonical form of finite subsets of $\Z^\MakeLowercase{d}$]{An algorithm for canonical forms of finite subsets of $\Z^d$ up to affinities}
\author{Giovanni Paolini}
\address{Scuola Normale Superiore, Piazza dei Cavalieri 7, 56126 Pisa (Italy)}
\email{giovanni.paolini@sns.it}

\begin{document}

\begin{abstract}
 In this paper we describe an algorithm for the computation of canonical forms of finite subsets of $\Z^d$, up to affinities over $\Z$. 
 For fixed dimension $d$, this algorithm has worst-case asymptotic complexity $O(n \log^2 n \, s\,\mu(s))$, where $n$ is the number of points in the given subset, $s$ is an upper bound to the size of the binary representation of any of the $n$ points, and $\mu(s)$ is an upper bound to the number of operations required to multiply two $s$-bit numbers.
 In particular, the problem is fixed-parameter tractable with respect to the dimension $d$.
 
 This problem arises e.g.\ in the context of computation of invariants of finitely presented groups with abelianized group isomorphic to $\Z^d$.
 In that context one needs to decide whether two Laurent polynomials in $d$ indeterminates, considered as elements of the group ring over the abelianized group, are equivalent with respect to a change of basis.
 
%  \subclass{MSC 52C07}
\end{abstract}

\maketitle

The final publication is available at Springer via \url{http://dx.doi.org/10.1007/s00454-017-9895-6}.

\section{Introduction}
\label{sec:introduction}

The problem we are going to study is that of algorithmically determining whether two configurations of points in the $d$-dimensional integral lattice can be obtained, one from the other, through an affine automorphism of $\Z^d$.
Such affine automorphism need not preserve the order in which the elements of the two sets are specified.
If we required instead the order to be preserved, the problem could be more easily solved using the Hermite normal form (see Section \ref{sec:canonical-form-with-frame}).

For instance, the second set of points in Figure \ref{fig:example} can be obtained from the first one by applying the affinity
\[ x\,\mapsto \m{-1 & -1\\ \phantom{-}0 & \phantom{-}1} x + \m{13\\ 0}. \]

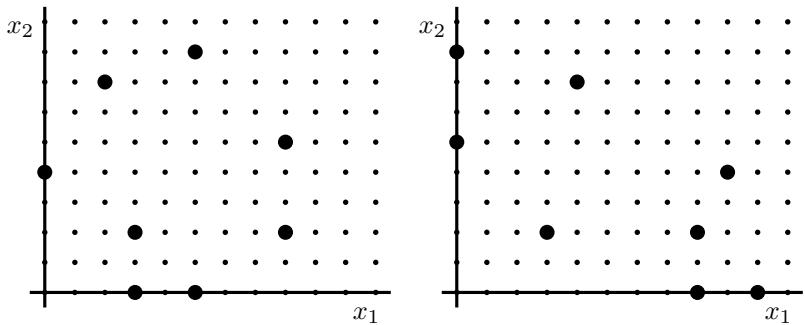
\begin{figure}[thbp]
  \begin{center}
    \begin{tikzpicture}[scale=0.4]
\draw[very thick] (-0.5,0) -- (11.5,0)
 node[left,xshift=0 cm,yshift=-0.3 cm]{$x_1$};
\draw[very thick] (0,-0.5) -- (0,9.5)
 node[left,yshift=-0.3 cm]{$x_2$};
%\draw[very thick,dashed] (0,4) -- (4,4);
\foreach \x in {0,...,11}
\foreach \y in {0,...,9}
\fill[black,xshift=\x cm,yshift=\y cm] (0,0) circle (0.09);
\fill[black,xshift=3 cm,yshift=0 cm] (0,0) circle (0.25);
\fill[black,xshift=5 cm,yshift=0 cm] (0,0) circle (0.25);
\fill[black,xshift=8 cm,yshift=2 cm] (0,0) circle (0.25);
\fill[black,xshift=8 cm,yshift=5 cm] (0,0) circle (0.25);
\fill[black,xshift=5 cm,yshift=8 cm] (0,0) circle (0.25);
\fill[black,xshift=2 cm,yshift=7 cm] (0,0) circle (0.25);
\fill[black,xshift=0 cm,yshift=4 cm] (0,0) circle (0.25);
\fill[black,xshift=3 cm,yshift=2 cm] (0,0) circle (0.25);
\end{tikzpicture}
    \begin{tikzpicture}[scale=0.4]
\draw[very thick] (-0.5,0) -- (11.5,0)
 node[left,xshift=0 cm,yshift=-0.3 cm]{$x_1$};
\draw[very thick] (0,-0.5) -- (0,9.5)
 node[left,yshift=-0.3 cm]{$x_2$};
%\draw[very thick,dashed] (0,4) -- (4,4);
\foreach \x in {0,...,11}
\foreach \y in {0,...,9}
\fill[black,xshift=\x cm,yshift=\y cm] (0,0) circle (0.09);
\fill[black,xshift=10 cm,yshift=0 cm] (0,0) circle (0.25);
\fill[black,xshift=8 cm,yshift=0 cm] (0,0) circle (0.25);
\fill[black,xshift=3 cm,yshift=2 cm] (0,0) circle (0.25);
\fill[black,xshift=0 cm,yshift=5 cm] (0,0) circle (0.25);
\fill[black,xshift=0 cm,yshift=8 cm] (0,0) circle (0.25);
\fill[black,xshift=4 cm,yshift=7 cm] (0,0) circle (0.25);
\fill[black,xshift=9 cm,yshift=4 cm] (0,0) circle (0.25);
\fill[black,xshift=8 cm,yshift=2 cm] (0,0) circle (0.25);
\end{tikzpicture}~~
  \end{center}
  \caption{Example of two equivalent sets for $d=2$.}
  \label{fig:example}
\end{figure}

Instead of trying to directly find if two given sets of points are equivalent (i.e. if there is an affinity that maps one to the other), we will describe a procedure to compute a ``canonical form'' of a set.
Then, to check the equivalence of two sets, it will suffice to check the equality of their canonical forms.

Different possible approaches to this problem can be tried, other than the one we present in this paper.
For instance one can exploit geometric and/or combinatorial constructions, such as the convex hull, which is equivariant under the action of the affine group.
However the optimal (non-output-sensitive) algorithm to compute the $d$-dimensional convex hull runs in time $O(n\log n + n^{\lfloor d/2 \rfloor})$ \cite{chazelle1993optimal}, which is slower than what we want to achieve.
We believe that in dimension $d \leq 3$ it is actually possible to devise a fast algorithm (linear up to a logarithmic factor) for our problem using the convex hull.

% Different possible approaches to this problem were tried by the author, before coming to the one described in this paper.
% For instance one can exploit geometric and/or combinatorial constructions, such as the convex hull, which is equivariant under the action of the affine group.
% In dimension $d\leq 3$, a fast algorithm (linear up to a logarithmic factor) based on the computation of the convex hull was indeed found, but was not published.
Another approach could be to define the canonical form of a set as an element of the orbit which minimizes some quantity (such as the $\|\cdot\|_1$ norm or the $\|\cdot\|_\infty$ norm), but we couldn't find a reasonably fast algorithm to do this.

The approach presented here is completely different, and is based on arithmetic properties of the integral lattice. Its advantages are the almost linear asymptotic complexity (in terms of the size of the set), the generality (it works for any dimension $d$, despite the running time strongly depends on it) and a simple implementation.

Following \cite{downey1994parameterized}, a problem is called \emph{fixed-parameter tractable} with respect to a
parameter $k \in \N$, if for every input with parameter less or equal to $k$, the problem can be solved in $O(f(k)\cdot n^{O(1)})$ time, where $f$ is an arbitrary function independent of the problem size $n$.
In terms of parametrized complexity, our results then imply that the problem we consider is fixed-parameter tractable with respect to the dimension $d$.

This problem fits into the theory of exact point-matching, where we replace here the usual $\R^d$ with the integral lattice $\Z^d$.
The algorithm we derive for affine integral point-matching is much faster compared with known algorithms for continuous point-matching.
For instance, congruency testing in $\R^d$ can be done in time $O(n^{d-2}\log n)$ for $d\geq 3$ with the algorithm described in \cite{alt1988congruence}.
It seems that arithmetic properties of integers play a fundamental role in making integral point-matching faster.

One situation in which this problem arises is in the context of isomorphism between finitely presented torsion-free groups, and particularly in the case of fundamental groups of topological spaces.
Recall that for a group $G$, its \emph{commutator subgroup} $G'$ is the subgroup generated by all commutators $[g,h] = ghg^{-1}h^{-1}$, and its \emph{abelianized group} is the quotient $G/G'$.
Let $G$ be a finitely presented group with abelianized group $H$ isomorphic to $\Z^d$, and let $\psi\colon H\to \Z^d$ be an isomorphism.
Suppose that $G$ admits a presentation with more generators than relations (e.g.\ this is the case for knot and link groups).
In a work of Fox \cite{fox1954calculus2} it is shown how to construct, from $G$ and $\psi$, a Laurent polynomial $\Delta(t_1,\dots,t_d)$, called \emph{Alexander polynomial}, which is defined up to a factor $\pm t_1^{\lambda_1}\cdots t_d^{\lambda_d}$ ($\lambda_i\in\Z$).
This polynomial depends on the chosen isomorphism $\psi$ between the abelianized group and $\Z^d$ (in other words, it depends on the choice of a basis for $H$).
In order to obtain an invariant of $G$ (up to group isomorphisms), one should determine a canonical form for the Alexander polynomial up to change of basis for $H$.
It turns out (see \cite{bellettini2015shape}, Sections 7.6, 7.7) that a change of basis, given by a linear automorphism $A$ of $\Z^d$, affects every monomial $\alpha t_1^{m_1}\cdots t_d^{m_d}$ by transforming the exponents vector $(m_1, \dots, m_d)^t$ with $A$.
Since the Alexander polynomial is itself determined up to a factor $\pm t_1^{\lambda_1}\cdots t_d^{\lambda_d}$, an invariant of $G$ is given by the Alexander polynomial up to the action of the group of affine automorphisms of $\Z^d$ (and a possible change of sign). The determination of a canonical form for such an action is therefore related to the problem we are going to discuss.

In Section \ref{sec:preliminaries} we formally state the problem we are going to solve, and we introduce some notation.
Throughout Sections \ref{sec:canonical-form-with-frame}-\ref{sec:canonical-form} we describe the algorithm, prove its correctness and analyze its complexity.
In Section \ref{sec:alexander} we recall the construction of the Alexander polynomial due to Fox, and we illustrate how to modify our algorithm in order to compute a canonical form of such polynomial.
Section \ref{sec:conclusion} provides a final discussion with some concluding remarks.

\section{Preliminaries and notations}
\label{sec:preliminaries}

Let $d$ be a fixed positive integer.
Let $\GL(d,\Z)$ be the group of linear automorphisms of $\Z^d$ over $\Z$, i.e. the group of $d\times d$ matrices with entries in $\Z$ and determinant $\pm 1$.
Moreover, let $\Aff(d,\Z)$ be the group of affinities of $\Z^d$.
The group $\Aff(d,\Z)$ can be regarded as a subgroup of $\GL(d+1,\Z)$: the affinity $x\mapsto Ax+b$, with $A\in\GL(d,\Z)$ and $b\in \Z^d$, is represented by the block matrix
\[ \m{ \,A\, & \,b\, \\ \,0\, & \,1\, }. \]

Let $X_d$ be the set of all finite subsets $\Lambda$ of $\Z^d$.
The natural action of $\Aff(d,\Z)$ on $\Z^d$ induces an action of $\Aff(d,\Z)$ on $X_d$: explicitly, if $\varphi\in \Aff(d,\Z)$ and $\Lambda = \{p_1, p_2, \dots, p_n\} \subseteq \Z^d$, the action is given by
\[ \varphi (\{p_1, p_2, \dots, p_n\}) = \{ \varphi(p_1), \varphi(p_2), \dots, \varphi(p_n)\}. \]

Our purpose is to describe a canonical form for elements of $X_d$ up to the action of $\Aff(d,\Z)$, and an algorithm for the computation of such a canonical form. For any fixed dimension $d$, our algorithm will have worst-case asymptotic complexity $O(n\log^2 n \, s\, \mu(s))$. Here, $n$ is the size of the given subset $\Lambda$ of $\Z^d$ (as above) and $s$ is an upper bound on the size of the binary representation of any coordinate of any point of $\Lambda$. Since $d$ is fixed, $s$ is also (up to a constant) an upper bound to the size of the binary representation of any point of $\Lambda$.
With $\mu(s)$ we indicate an upper bound to the cost of multiplying two $s$-bit integers; for instance, using the Schönhage-Strassen algorithm \cite{schonhage1971schnelle} we would have $\mu(s) = O(s\,\log s \,\log\log s)$.

Since the concept of ``canonical form'' plays a key role in this work, we give the following formal definition.
\begin{definition}
  \label{def:canonical-form}
  Let $S$ be a set, and $G$ be a group acting on $S$. A canonical form for $S$ with respect to the action of $G$ is a function $f\colon S \to S$ satisfying the following two conditions:
  \begin{enumerate}[(1)]
    \item $f(x) \in \Orb(x)$ for all $x\in S$ (here we denote by $\Orb(x)$ the orbit of $x$);
    \item $f(gx) = f(x)$ for all $x\in S$ and $g\in G$.
  \end{enumerate}
  We also say that $f\colon S\to S$ is a weak canonical form if it satisfies condition (2) but does not necessarily satisfy condition (1).
\end{definition}
The second condition simply says that $f$ is constant over any orbit, so $f$ picks a ``canonical representative'' from each orbit.
Having a computable canonical form allows to test whether two elements $x,y\in S$ belong to the same orbit: this happens if and only if $f(x)=f(y)$.

We now give a few more definitions which will be useful later.

\begin{definition}
  A \emph{frame} is an ordered set of affinely independent points of $\Q^d$.
  Given a set $\Lambda\subseteq \Z^d$, a \emph{$\Lambda$-frame} is a frame included in $\Lambda$.
  A frame $Q$ is \emph{$\Lambda$-covering} if $\Lambda\subseteq\Span(Q)$. A $\Lambda$-frame which is also $\Lambda$-covering is shortly called a \emph{complete $\Lambda$-frame}.
\end{definition}

By ``$\Span$'', we always mean the affinely generated subspace over the field of the rational numbers (not over $\Z$). Also the expression ``affinely independent'' is always to be intended over $\Q$, not over $\Z$.
Notice that a $\Lambda$-covering frame is not necessarily a $\Lambda$-frame.

Let $Y_d$ be the set of the pairs $(\Lambda,Q)$, with $\Lambda\in X_d$ and $Q$ a $\Lambda$-covering frame.
Roughly speaking, an element of $Y_d$ is a finite subset of $\Z^d$ together with an affine coordinate system.
In order to find a canonical form for elements $\Lambda\in X_d$, we will first do it for elements $(\Lambda,Q)\in Y_d$, and then we will describe a canonical way to choose a frame $Q$ for each set $\Lambda$.

Finally, by ``lexicographic order'' we will mean the following.
\begin{itemize}
  \item For elements of $\Q^d$, this is the usual lexicographic order.
  \item When comparing two finite \emph{ordered} sets of elements of $\Q^d$ (for instance, two frames), first compare their size and then (if the two sets have the same size) compare them lexicographically.
  \item When comparing two finite \emph{unordered} sets of elements of $\Q^d$ (for instance, two elements of $X_d$), sort each of them lexicographically, and then compare them as ordered sets.
  \item When comparing two elements $(\Lambda_1, Q_1)$ and $(\Lambda_2, Q_2)$ of $Y_d$, first compare $\Lambda_1$ and $\Lambda_2$ and then (if $\Lambda_1 = \Lambda_2$) compare $Q_1$ and $Q_2$.
\end{itemize}

\section{Canonical form, given a frame}
\label{sec:canonical-form-with-frame}

In this section, we want to describe an algorithm that, given a pair $(\Lambda,Q)\in Y_d$, returns a pair $(\Omega,U) = f(\Lambda,Q)$ which has the following properties:
\begin{enumerate}
  \item $\Omega$ is a finite subset of $\Z^d$ in the orbit of $\Lambda$ with respect to the action of $\Aff(d,\Z)$;
%   such that there is an integer affinity that sends $\Lambda$ to $\Omega$;
  \item $U$ is a complete $\Omega$-frame;
  \item $f$ is a weak canonical form for $Y_d$ in the sense of Definition \ref{def:canonical-form}, i.e.
  \[ f \big(\varphi(\Lambda), \varphi(Q)\big) = f(\Lambda,Q)\quad \forall\,\varphi\in\Aff(d,\Z). \]
\end{enumerate}
Notice that $f$ is not required to be a canonical form (the frame $U$ is not necessarily the image of the frame $Q$ under some affinity; in general, they don't even have the same size).

In what follows we are going to use the Hermite normal form (see for instance \cite{cohen1993course} and \cite{newman1972integral}), shortened ``HNF''. The Hermite normal form of an integral $d\times n$ matrix is a canonical form up to left-multiplication by elements of $\GL(d,\Z)$, satisfying the following additional properties.
\begin{itemize}
  \item It is an upper triangular $d\times n$ matrix, and zero rows are located below non-zero rows.
  \item The pivot (i.e.\ the first non-zero entry) of a non-zero row is positive, and is strictly to the right of the pivot of the row above it.
  \item The elements below pivots are zero, and elements above pivots are non-negative and strictly smaller than the pivot.
\end{itemize}

The following is the pseudocode of the algorithm, which will be subsequently described in words.

\begin{algorithm}[Canonical form, given a frame]
  \label{alg:canonical-form-with-frame}
  This algorithm takes as input a pair $(\Lambda,Q)\in Y_d$ and outputs a pair $(\Omega,U)$ with the properties described above.
  
  \begin{alg}
    \Function{CanonicalFormWithFrame}{$\Lambda,Q$}
      \State $T \gets Q\cap \Lambda$ as a list, with the ordering induced by $Q$ \label{line:initialize-T}
      \State $k \gets \dim\Span(\Lambda)$
      \While{$|T| < k+1$} \label{line:frame-extraction-begin}
	\State $p \gets $ the point of $\Lambda \setminus \Span(T)$ such that its coordinates with respect \label{line:add-point1}
	\State \qquad to the frame $Q$ are lexicographically minimal \label{line:add-point2}
	\State $T \gets T \cup \{p\}$
      \EndWhile \label{line:frame-extraction-end}
      \State $\{p_0,\dots, p_k\} \gets T$
      \State $M \gets d\times k$ matrix with columns $p_1-p_0,\, \dots,\, p_k-p_0$ \label{line:definition-M}
      \State $A \gets$ any element of $\GL(d,\Z)$ such that $AM=\Call{HNF}{M}$ \label{line:hermite}
      \State $\psi \gets$ affinity defined by $x\mapsto A(x-p_0)$ \label{line:definition-psi}
      \State \Return $(\Omega,U) = \big(\psi(\Lambda), \psi(T)\big)$
    \EndFunction
  \end{alg}
\end{algorithm}

Let us describe briefly the steps of Algorithm \ref{alg:canonical-form-with-frame}.
In line \ref{line:initialize-T}, we initialize a new frame $T$ (which is actually a $\Lambda$-frame), extracting from $Q$ the elements that also belong to $\Lambda$. $T$ is given the ordering induced as a subset of $Q$.
Then, in lines \ref{line:frame-extraction-begin}-\ref{line:frame-extraction-end}, we complete $T$ to a $\Lambda$-covering frame using points of $\Lambda$ (thus $T$ becomes a complete $\Lambda$-frame). This is done adding a point of $\Lambda$ at a time, each time choosing the point that is lexicographically minimal with respect to the frame $Q$.
Then, if $p_0,\dots,p_k$ are the elements of $T$, in line \ref{line:definition-M} we define the matrix $M$ as
\[ M = \Big(\, p_1-p_0 \; \Big| \; \dots \; \Big| \; p_k-p_0 \, \Big). \]
In line \ref{line:hermite}, we define $A$ as any $d\times d$ matrix such that the left multiplication by $A$ sends $M$ to its Hermite normal form (the algorithm in \cite{storjohann1996asymptotically} computes both the HNF and such an $A$).
Finally, we define $\psi$ as the affinity $x\mapsto A(x-p_0)$, which is the affinity that maps $p_0$ to the origin and each $p_i$ ($i=1,\dots,k$) to the $i$-th column of $\Call{HNF}{M}$. The affinity $\psi$ is used to transform the pair $(\Lambda,T)$ into the pair which is then returned.
Properties 1 and 2 at the beginning of this section are automatically verified by Algorithm \ref{alg:canonical-form-with-frame}. Property 3 is given by the following theorem.

\begin{theorem}
  \label{thm:correctness-canonical-form-with-frame}
  The function defined by Algorithm \ref{alg:canonical-form-with-frame} is a weak canonical form.
\end{theorem}

\begin{proof}
  Suppose the function \Call{CanonicalFormWithFrame}{} is given as input the pair $(\tilde\Lambda, \tilde Q)$ instead of $(\Lambda,Q)$, where $\tilde\Lambda = \varphi(\Lambda)$ and $\tilde Q = \varphi(Q)$ for some $\varphi\in\Aff(d,\Z)$.
  Let us analyze how this change affects the output. We denote all the variables of the execution of the call \Call{CanonicalFormWithFrame}{$\tilde\Lambda,\tilde Q$} by adding a tilde over them, in order to distinguish them from those of the call \Call{CanonicalFormWithFrame}{$\Lambda,Q$}.
  
  First (line \ref{line:initialize-T}), we have $\tilde T = \tilde Q \cap \tilde\Lambda = \varphi(Q) \cap \varphi(\Lambda) = \varphi(T)$.
  Then we turn to lines \ref{line:frame-extraction-begin}-\ref{line:frame-extraction-end}: inductively it is easy to show that, at each step of the loop, $\tilde T = \varphi(T)$ (this is true because the coordinates of a point $p\in \Lambda$ with respect to the frame $Q$ are the same as the coordinates of the point $\varphi(p)\in\varphi(\Lambda)$ with respect to the frame $\varphi(Q)$).
  So, after the execution of the loop, we still have $\tilde T = \varphi(T)$ as ordered sets.
  Let $\tilde p_0,\dots \tilde p_k$ be the elements of $\tilde T$, so that $\tilde p_i = \varphi(p_i)$ for all $i$.
  Let $B \in \GL(d,\Z)$ be the linear part of the affinity $\varphi$ (so $\varphi$ has the form $x\mapsto Bx+v$, for some $v\in\Z^d$). The $i$-th column of the matrix $\tilde M$ (line \ref{line:definition-M}) is then $\tilde p_i-\tilde p_0 = \varphi(p_i) - \varphi(p_0) = B(p_i-p_0)$. Thus we have the relation $\tilde M = BM$, which means that the matrices $\tilde M$ and $M$ are equivalent to each other, up to left multiplication. In particular, they have the same Hermite normal form.
  This means that $\tilde A \tilde M = AM$ (line \ref{line:hermite}).
  In other words (reading this equality column by column), for each $i$ we have that
  \begin{equation}
    \label{eqn:transformation}
    \tilde A \big( B(p_i-p_0) \big) = A (p_i-p_0).
  \end{equation}
  Relation \eqref{eqn:transformation} can be interpreted as follows: the linear transformations $\tilde A B$ and $A$ coincide on the vectors $p_i-p_0$, and so they coincide on the linear span of these vectors, which is the linear subspace parallel to $\Span(\Lambda)$.
  The affinity $\tilde\psi$ defined in line \ref{line:definition-psi} maps $x$ to $\tilde A (x - \tilde p_0) = \tilde A (x - \varphi(p_0))$. Now, let $\tilde p = \varphi(p)$ be a point of $\tilde\Lambda$, with $p\in\Lambda$. Then,
  \begin{IEEEeqnarray*}{rCl}
    \tilde\psi(\tilde p) &=& \tilde A \big( \tilde p - \varphi(p_0) \big) = \tilde A \big( \varphi(p) - \varphi(p_0) \big) \\
    &=& \tilde A \big( B(p-p_0) \big) \stackrel{(*)}{=} A(p-p_0) = \psi(p),
  \end{IEEEeqnarray*}
  where the equality marked with a $(*)$ follows by the fact that $p-p_0$ belongs to the linear subspace parallel to $\Span(\Lambda)$.
  So we have finally proved that $\tilde\psi(\tilde\Lambda) = \psi(\Lambda)$.
  Similarly, we also have that $\tilde\psi(\tilde p_i) = \psi(p_i)$ for all $i$, and thus $\tilde\psi(\tilde T) = \psi(T)$ as ordered sets. \qed
\end{proof}

\begin{remark}
  Line \ref{line:initialize-T} of Algorithm \ref{alg:canonical-form-with-frame} could be replaced by the simpler initialization ``$T\gets \varnothing$'', and Theorem \ref{thm:correctness-canonical-form-with-frame} would still hold.
  But the slightly more complicated initialization of line \ref{line:initialize-T} will be required in the proof of Lemma \ref{lemma:limited}.
\end{remark}

\begin{theorem}
  \label{thm:complexity-canonical-form-with-frame}
  If $d$ is fixed, then the worst-case asymptotic complexity of Algorithm \ref{alg:canonical-form-with-frame} is $O(n\, \mu(s))$, where $n=|\Lambda|$ and $s$ is an upper bound on the size of the binary representation of any element of $\Lambda$ or $Q$.
  Moreover, the size of the binary representation of any element of the returned set $\Omega$ is $O(s)$.
\end{theorem}

\begin{proof}
  Consider lines \ref{line:add-point1}-\ref{line:add-point2}.
  Recognizing if a point $p\in \Lambda$ belongs to $\Span(T)$ reduces to compute the determinant of a $(d+1)\times (d+1)$ matrix. The product of any $d+1$ entries of this matrix has size $O((d+1)s) = O(s)$.
  Therefore the computation of the determinant requires a constant number of sums and products of numbers of size $O(s)$, i.e.\ $\mu(O(s)) = O(\mu(s))$ operations (exchanging $\mu$ and $O$ can be done because $\mu$ is polynomial).
  
  Finding the coordinates of $p$ with respect to the frame $Q$ can be done inverting a $(d+1)\times (d+1)$ matrix, which reduces to computing a constant number of determinants. Thus, as before, $O(\mu(s))$ operations are needed and the obtained coordinates have size $O(s)$.
  
  Then lines \ref{line:add-point1}-\ref{line:add-point2} require $O(n\, \mu(s))$ operations, because for each point of $\Lambda$ we need to perform the above operations and then possibly update the minimal coordinates found so far.
  There are at most $d+1 = O(1)$ iterations of the while loop, so lines \ref{line:initialize-T}-\ref{line:definition-M} require $O(n\, \mu(s))$ operations.
  
  Since the size of the binary representation of $M$ is $O(s)$, the computation of the Hermite normal form of $M$ can be done in $O(\mu(s))$ time, for instance with the algorithm in \cite{storjohann1996asymptotically}.
%   combining the triangularization algorithm in \cite{hafner1991asymptotically} and the HNF computation in \cite{storjohann1998computing}.
  In \cite{storjohann1998computing} it is also shown that $\log \| \Call{HNF}{M} \| = O(s)$, where $\|\cdot\|$ denotes the maximum absolute value of an entry of the matrix.
  In $O(\mu(s))$ time, $A$ and $\psi$ can be computed too (and they share the same bound on the coefficients as \Call{HNF}{$M$}).
  Finally, the computation of $\psi(\Lambda)$ requires $O(n\,\mu(s))$ operations, and that of $\psi(T)$ requires $O(\mu(s))$ operations.
  
  The fact that the binary representation of the elements of $\Omega$ is $O(s)$ follows easily from the previous arguments. \qed
\end{proof}

\section{Getting an \texorpdfstring{$\Aff(d,\Z)$}{Aff(d,Z)}-equivariant set of complete \texorpdfstring{$\Lambda$}{Lambda}-frames}
\label{sec:canonical-frames}

We will now describe an algorithm which, given an input set $\Lambda\subseteq\Z^d$, equivariantly returns a nonempty set of complete $\Lambda$-frames. Here by ``equivariantly'' we mean $\Aff(d,\Z)$-equivariantly: if $\tilde\Lambda = \varphi(\Lambda)$ for some $\varphi\in\Aff(d,\Z)$, and the output of the algorithm applied to $\Lambda$ is a set of frames $\{R_1, \dots, R_m\}$, then the output of the algorithm applied to $\tilde\Lambda$ is $\{ \varphi(R_1), \dots, \varphi(R_m)\}$.
Notice that the output set of frames is unordered, whereas each of the frames is itself an ordered set of points.

\begin{remark}
  The set of all complete $\Lambda$-frames satisfies the equivariance condition, but this set is too large for any practical purpose (its size is $\Omega(n^d)$ in the worst case).
  For this reason we need to equivariantly select a ``small'' subset of it.
\end{remark}

From now on we will denote by $\F(\Lambda)$ the set of all $\Lambda$-frames, and by $\F_c(\Lambda)\subseteq\F(\Lambda)$ the set of all complete $\Lambda$-frames.
% Recall from Section \ref{sec:preliminaries} that we denote by $\F(\Lambda)$ the set of all the $\Lambda$-frames, and by $\F_c(\Lambda)$ the set of all complete $\Lambda$-frames.
The pseudocode for the above-mentioned algorithm is the following.

\begin{algorithm}[Equivariant set of complete $\Lambda$-frames]
  \label{alg:canonical-frames}
  \ \\
  This algorithm takes as input a finite set $\Lambda \subseteq \Z^d$ and equivariantly returns a nonempty set of complete $\Lambda$-frames.
  \begin{alg}
    \Function{EquivariantFrames}{$\Lambda$}
      \If{$|\Lambda| = 1$} \Comment{$\Lambda$ has size $1$ (base step of the recursion)}
	\State \Return $\{\Lambda\}$
      \EndIf
      \If{all the points of $\Lambda$ are congruent modulo $2$} \label{line:congruence-condition}
	\State $p \gets $ any point of $\Lambda$ \label{line:choice-of-p}
	\State $\Lambda' \gets (\Lambda - p) / 2$ \Comment{translate point $p$ to origin and divide by $2$} \label{line:shrink}
	\State $S \gets $ \Call{EquivariantFrames}{$\Lambda'$} \label{line:first-call}
	\State \Return $2S + p$ \Comment{perform the inverse transformation} \label{line:return-frames}
      \Else
	\State $\{\Lambda_1,\dots,\Lambda_h\} \gets$ partition of $\Lambda$ in congruence classes modulo 2 \label{line:partition}
	\For{$i\in \{1,\dots, h\}$}
	  \State $S_i \gets$ \Call{EquivariantFrames}{$\Lambda_i$}
	\EndFor
	\State $\Lambda' \gets \bigcup_{i=1}^h \bigcup_{R\in S_i} R$
	\State $F \gets \{ \, T \in\F_c(\Lambda') \,\mid\, T$ lexicographically minimizes \label{line:frames1}
	\State \qquad  $\Call{CanonicalFormWithFrame}{\Lambda,T}$ on $\F_c(\Lambda') \}$ \label{line:frames2}
	\State \Return $F$
      \EndIf
    \EndFunction
  \end{alg}
\end{algorithm}

We will now prove the correctness of Algorithm \ref{alg:canonical-frames}.

\begin{lemma}
  \label{lemma:mod-2}
  Let $p,q\in \Z^d$, and let $\varphi$ be an element of $\Aff(d,\Z)$.
  Then, $p\equiv q \pmod 2$ if and only if $\varphi(p) \equiv \varphi(q) \pmod 2$.
  In other words, being congruent modulo $2$ is an affine invariant.
\end{lemma}

\begin{proof}
  The affinity $\varphi$ will be of the form $x\mapsto Ax+b$, for some $A\in\GL(d,\Z)$ and $b\in\Z^d$.
  The condition $p\equiv q \pmod 2$ is equivalent to $p=q+2v$ for some $v\in \Z^d$.
  Applying $\varphi$ to both sides we obtain $\varphi(p) = \varphi(q+2v) = \varphi(q) + 2Av$, so $\varphi(p)\equiv\varphi(q) \pmod 2$.
  The same argument applied to $\varphi^{-1}$ proves the converse implication. \qed
\end{proof}

\begin{theorem}
  \label{thm:correctness-canonical-frames}
  Algorithm \ref{alg:canonical-frames} terminates, and it equivariantly returns a set of complete $\Lambda$-frames.
\end{theorem}

\begin{proof}
  In each recursive call of \Call{EquivariantFrames}{}, either the diameter of $\Lambda$ is halved or the size of $\Lambda$ decreases: the former case occurs if the condition of line \ref{line:congruence-condition} holds; the latter occurs in the other case (if not all the points of $\Lambda$ are congruent modulo $2$, then the partition of line \ref{line:partition} is made of at least two non-empty subsets).
  As long as the size of $\Lambda$ remains greater than $1$, the diameter of $\Lambda$ is $\geq 1$.
  So neither of the two above possibilities can happen infinitely many times, and the algorithm eventually terminates.
  
  We prove the rest of the statements by induction on the size $n$ of $\Lambda$ and its diameter.
  The base step is $n=1$, for which the claim is obvious.
  We may now assume $n>1$. We will distinguish two cases.
  \begin{itemize}
    \item First case: the condition of line \ref{line:congruence-condition} holds.
    Clearly, by induction, the frames returned in line \ref{line:return-frames} are complete $\Lambda$-frames.
    Suppose now that $\Lambda$ is replaced by $\tilde\Lambda = \varphi(\Lambda)$, where $\varphi$ is the affinity given by $x\mapsto Ax+b$. In line \ref{line:choice-of-p}, a point $\tilde p = \varphi(q)$ is chosen, for some $q\in\Lambda$.
    The set $\tilde\Lambda'$ calculated in line \ref{line:shrink} is given by
    \begin{IEEEeqnarray*}{rCl}
      \tilde\Lambda' &=& \frac{\tilde\Lambda - \tilde p}{2} = \frac{\varphi(\Lambda) - \varphi(q)}{2} = \frac{A(\Lambda) - A(q)}{2} \\
      &=& A \left(\frac{\Lambda-p}{2}\right) + A \left(\frac{p-q}{2}\right) = A (\Lambda') + A\left(\frac{p-q}{2}\right).
    \end{IEEEeqnarray*}
    Then the sets $\Lambda'$ and $\tilde\Lambda'$ can be obtained one from the other by applying the affinity $x\mapsto Ax + A\left(\frac{p-q}{2}\right)$.
    Here notice that $\frac{p-q}{2}\in\Z^d$, so the constructed transformation is really an affinity that preserves the lattice $\Z^d$.
    By induction, the set of frames $\tilde S$ calculated in line \ref{line:first-call} is equivariant.
    Consequently,
    \[ \tilde S = A(S) + A\left(\frac{p-q}{2}\right). \]
    Finally, the return value of the call \Call{EquivariantFrames}{$\tilde\Lambda$} is
    \begin{IEEEeqnarray*}{rCl}
      2\tilde S + \tilde p &=& 2A(S) + 2A\left(\frac{p-q}{2}\right) + \varphi(q) \\
      &=& 2A(S) + A(p) - A(q) + A(q) + b \\
      &=& A( 2S + p ) + b \\
      &=& \varphi( 2S+p ),
    \end{IEEEeqnarray*}
    which is what we wanted.
    
    \item Second case: the condition of line \ref{line:congruence-condition} does not hold.
    The sets in $F$ are complete $\Lambda'$-frames. By construction, $\Lambda'$ is a subset of $\Lambda$ and $\Span(\Lambda') = \Span(\Lambda)$. So the sets in $F$ are also complete $\Lambda$-frames.
    By Lemma \ref{lemma:mod-2}, the partition computed in line \ref{line:partition} is $\Aff(d,\Z)$-equivariant (notice that such partition is an unordered set). Thus the unordered set $\{S_1, \dots, S_h\}$ is also equivariant, and so is the union $\Lambda'$.
    Finally minimizing \Call{CanonicalFormWithFrame}{$\Lambda, T$} is an equivariant condition since \Call{CanonicalFormWithFrame}{$\Lambda, T$} is $\Aff(d,\Z)$-invariant (by Theorem \ref{thm:correctness-canonical-form-with-frame}), so $F$ is itself equivariant. \qed
  \end{itemize}
\end{proof}

Unfortunately, when the partition found in line \ref{line:partition} of Algorithm \ref{alg:canonical-frames} is very unbalanced (for instance, if $|\Lambda_1| = 1$ and $|\Lambda_2|=n-1$), the depth of the tree of the recursive calls can grow linearly with $n$.
Then the overall complexity can be quadratic in $n$, as each recursive call takes linear time.
This is already good compared with the $O(n^d)$ trivial algorithm, but we are going to present a variant to Algorithm \ref{alg:canonical-frames} with almost linear worst-case asymptotic complexity (as claimed in Section \ref{sec:preliminaries}).
The idea is to change what is done in lines \ref{line:frames1}-\ref{line:frames2}, which is a quite rough way to find an equivariant set of frames.
To do this, the new recursive function \Call{EquivariantFrames2}{$\Lambda, Q$} takes one more argument, a frame $Q$, and equivariantly returns a nonempty set $S$ of $\Lambda$-frames such that, for each $R\in S$, $Q\cup R$ is a $\Lambda$-covering frame.
Similarly to Algorithm \ref{alg:canonical-frames}, here by ``equivariantly'' we mean that, for any $\varphi\in\Aff(d,\Z)$, we have
\[ \Call{EquivariantFrames2}{\varphi(\Lambda), \varphi(Q)} = \varphi\big( \Call{EquivariantFrames2}{\Lambda,Q} \big). \]

In what follows, we say that a partition $\{\Lambda_1,\dots,\Lambda_h\}$ of $\Lambda$ is \emph{balanced} if $|\Lambda_i|\leq n/2$ for all $i$, where $n$ is the size of $\Lambda$. Otherwise, we say that the partition is \emph{unbalanced}.

\begin{algorithm}[Equivariant set of $\Lambda$-frames]
  \label{alg:canonical-frames2}
  This algorithm takes as input a finite set $\Lambda \subseteq \Z^d$ and a frame $Q\subseteq \Q^d$, and equivariantly returns a nonempty set $S$ of $\Lambda$-frames, with the property that $Q\cap R=\varnothing$ and $Q\cup R$ is a $\Lambda$-covering frame for each $R\in S$.
  \begin{alg}
    \Function{EquivariantFrames2}{$\Lambda, Q$}
      \State $\Lambda \gets \Lambda \setminus \Span(Q)$ \label{line:begin}
      \If{$|\Lambda| \leq 1$}
	\State \Return $\{\Lambda\}$ \label{line:return-base-case}
      \EndIf \label{line:end-first-if}
      \If{all the points of $\Lambda$ are congruent modulo $2$} \label{line:congruence-condition2}
	\State $p \gets $ any point of $\Lambda$ \label{line:first-case-begin}
	\State $\Lambda' \gets (\Lambda - p) / 2$
	\State $Q' \gets (Q - p) / 2$
	\State $S' \gets $ \Call{EquivariantFrames2}{$\Lambda', Q'$} \label{line:first-call2}
	\State $S \gets 2S' + p$
	\State \Return $S$ \label{line:first-case-end}
      \Else
	\State $\{\Lambda_1,\dots,\Lambda_h\} \gets$ partition of $\Lambda$ in congruence classes modulo 2 \label{line:second-case-begin}
	\State sort $\{\Lambda_1,\dots,\Lambda_h\}$ by size \Comment{after this, we have $|\Lambda_1|\leq \dots\leq |\Lambda_h|$} \label{line:sort}
	\If{$|\Lambda_h| \leq n/2$} \Comment{the partition is balanced}
	  \For{$i\in \{1,\dots, h\}$} \label{line:balanced-begin}
	    \State $S_i \gets$ \Call{EquivariantFrames2}{$\Lambda_i, Q$} \label{line:second-call2}
	  \EndFor
	  \State $\Lambda' \gets \bigcup_{i=1}^h \bigcup_{V \in S_i} V$ \label{line:big-union}
	  \State $F \gets \{ \, R\in\F(\Lambda') \,\mid\, Q\cap R =\varnothing, \; Q\cup R \text{ is a $\Lambda$-covering frame,}$ \label{line:F-first-assignment}
	  \State \qquad and \Call{CanonicalFormWithFrame}{$\Lambda,Q\cup R$} is
	  \State \qquad lexicographically minimal $\}$ \label{line:balanced-end}
	\Else \Comment{the partition is unbalanced}
	  \For{$i\in \{1,\dots, h-1\}$} \label{line:unbalanced-begin}
	    \State $S_i \gets$ \Call{EquivariantFrames2}{$\Lambda_i, Q$} \label{line:third-call2}
	  \EndFor
	  \State $\Lambda' \gets \bigcup_{i=1}^{h-1} \bigcup_{V \in S_i} V$ \label{line:lambda'}
	  \State $E \gets \{ \, R\in\F(\Lambda') \,\mid\, Q\cap R=\varnothing \text{ and } Q\cup R \text{ is a }$ \label{line:E-beginning}
	  \State \qquad $(\Lambda\setminus\Lambda_h)$-covering frame $\}$\label{line:E}
	  \ForAll{$R\in E$}
	    \State $S_R \gets \Call{EquivariantFrames2}{\Lambda_h, Q \cup R}$ \label{line:fourth-call2}
	  \EndFor
	  \State $F \gets \{ R \cup T \in \F(\Lambda') \mid R\in E,\, T\in S_R, \text{ and}$ \label{line:F-second-assignment}
	  \State \qquad $\Call{CanonicalFormWithFrame}{\Lambda, Q\cup R\cup T} \text{ is}$
	  \State \qquad lexicographically minimal $\}$ \label{line:unbalanced-end}
	\EndIf
	\State \Return $F$ \label{line:second-case-end}
      \EndIf
    \EndFunction
  \end{alg}
\end{algorithm}

Let us describe briefly the steps of Algorithm \ref{alg:canonical-frames2}.
As in Algorithm \ref{alg:canonical-frames}, the main distinction is given by the condition in line \ref{line:congruence-condition2} (whether all the points of $\Lambda$ are congruent modulo 2, or not).

If the condition in line \ref{line:congruence-condition2} holds (lines \ref{line:first-case-begin}-\ref{line:first-case-end}), all the points of $\Lambda$ and $Q$ are translated and their coordinates are then divided by 2. The recursive call in line \ref{line:first-call2} gives an unordered set $S'$ of $\Lambda'$-frames, which is then transformed to a set $S$ of $\Lambda$-frames.

Lines \ref{line:second-case-begin}-\ref{line:second-case-end} are executed if the condition in line \ref{line:congruence-condition2} is not verified.
If the partition of line \ref{line:second-case-begin} is balanced (lines \ref{line:balanced-begin}-\ref{line:balanced-end}), the behaviour is similar to that of Algorithm \ref{alg:canonical-frames}: the function \Call{EquivariantFrames2}{} is called recursively for each subset of the partition (line \ref{line:second-call2}), and then the results are put together in lines \ref{line:big-union}-\ref{line:balanced-end}.
If the partition is unbalanced (lines \ref{line:unbalanced-begin}-\ref{line:unbalanced-end}), the largest subset ($\Lambda_h$) is treated separately (notice that the loop in line \ref{line:unbalanced-begin} ranges from 1 to $h-1$, not from 1 to $h$): as we will see, the reason for this is to increase the asymptotic efficiency.

\begin{remark}
  If $Q=\varnothing$, then Algorithm \ref{alg:canonical-frames2} equivariantly returns a set of complete $\Lambda$-frames, exactly as Algorithm \ref{alg:canonical-frames} does. However, the two outputs may differ: Algorithms \ref{alg:canonical-frames} and \ref{alg:canonical-frames2} calculate different ``equivariant forms''.
\end{remark}

\begin{remark}
  The number of frames $R$ considered in lines \ref{line:F-first-assignment}-\ref{line:balanced-end} (and, similarly, in lines \ref{line:E-beginning}-\ref{line:E}) can be considerably reduced in many ways.
  For instance, one can consider frames obtained as a concatenation of frames in the union of the $S_i$'s.
  Plus, if it is possible to distinguish the $\Lambda_i$'s in some way (e.g.\ because they have different size, or because the $S_i$'s have different size), one can reduce the number of concatenations to consider.
  However this does not affect the asymptotic complexity for fixed $d$, so we chose to present a simpler and more naive approach.
\end{remark}

In the rest of this section we will prove the correctness of Algorithm \ref{alg:canonical-frames2} and we will analyze its asymptotic complexity.

\begin{theorem}
  \label{thm:correctness-canonical-frames2}
  Algorithm \ref{alg:canonical-frames2} terminates, and the set of frames it returns is equi\-va\-riant.
\end{theorem}

\begin{proof}
  The arguments are similar to those of Theorem \ref{thm:correctness-canonical-frames}.
  The only main dif\-fe\-rence is given by how the set $\Lambda_h$ is treated separately in lines \ref{line:unbalanced-begin}-\ref{line:unbalanced-end}: however, since there can only be one subset of size greater than $n/2$, if such a subset exists then it is equivariant (because the partition itself is equivariant, as we already pointed out in the proof of Theorem \ref{thm:correctness-canonical-frames}). So, treating $\Lambda_h$ differently from the other subsets of the partition does not violate the equivariance of the procedure. \qed
\end{proof}

\begin{lemma}
  \label{lemma:limited}
  For each dimension $d>0$ there exists $h_d>0$ such that, for any $\Lambda'\subseteq \Lambda\subseteq \Z^d$ and frame $Q$ such that $\Lambda \subseteq \Span(Q\cup \Lambda')$, and for any pair $(\Omega, U)\in Y_d$, the set
  \begin{IEEEeqnarray*}{rCl}
    \mathcal{S} = \{ R \in \F(\Lambda') &\,\mid\,& Q\cap R = \varnothing, \; Q\cup R \text{ is a $\Lambda$-covering frame, and} \\ && \Call{CanonicalFormWithFrame}{\Lambda,Q\cup R} = (\Omega,U) \}
  \end{IEEEeqnarray*}
  has size $\leq h_d$.
\end{lemma}

\begin{proof}
  Let $n=|\Lambda|$ and $k = \dim\Span(\Lambda)$.
  First of all, we reduce to the case $\Span\Lambda = \Q^k$ (where $\Q^k$ is the affine subspace of $\Q^d$ consisting of the points with the last $d-k$ coordinates equal to zero).
  Notice that, if we change both $\Lambda$ and $Q$ by an affinity $\varphi\in\Aff(d,\Z)$, then $\mathcal{S}$ changes also by $\varphi$, and so its size remains the same.
  Let us choose the affinity $\varphi$ (of the form $x\mapsto A(x+b)$) as follows.
  \begin{itemize}
   \item $b=-q$, where $q$ is any fixed point of $\Lambda$.
   \item Let $M$ be the $d\times n$ matrix with columns given by the vectors $p-q$, for $p\in\Lambda$ (the columns can be arranged in any order).
   Then, let $A\in \GL(d,\Z)$ be such that $AM = \Call{HNF}{M}$.
  \end{itemize}
  Since the rank of $M$ is $k$, the Hermite normal form of $M$ has the last $d-k$ rows equal to zero. So the image of $\Lambda$ through the affinity $\varphi$ is included in $\Q^k$, as we wished.
  We can thus assume, from now on, that $\Span(\Lambda)=\Q^k$.
  
  Let $G$ be the subgroup of $\Aff(k,\Z)$ given by the affinities $\vartheta$ of $\Z^k\subseteq \Q^k$ such that $\vartheta(\Lambda)=\Lambda$. Since $\Lambda$ generates $\Q^k$ (as an affine space over $\Q$), an affinity $\vartheta\in G$ is completely determined by its restriction to $\Lambda$. Such a restriction is a permutation of $\Lambda$, by definition of $G$. So the order of $G$ is at most $n!$, and in particular $G$ is finite.
  
  We are now going to define an injective map $\chi\colon \mathcal{S} \to G$.
  Fix a frame $R_0\in \mathcal{S}$. Let $R$ be any frame in $\mathcal{S}$. By definition of $\mathcal{S}$, we have
  \[ \Call{CFWF}{\Lambda,Q\cup R} = \Call{CFWF}{\Lambda,Q\cup R_0} = (\Omega,U), \]
  where we have shortened ``\Call{CanonicalFormWithFrame}{}'' with ``\Call{CFWF}{}''.
  Let $T$ and $T_0$ be the complete $\Lambda$-frames constructed throughout the execution of the function \Call{CFWF}{} applied to $(\Lambda,Q\cup R)$ and $(\Lambda,Q\cup R_0)$, respectively. Recall that, as an immediate consequence of how the function \Call{CFWF}{} is defined, there exist (not necessarily unique) affinities $\psi,\psi_0\in\Aff(d,\Z)$ such that
  \[ \psi(\Lambda)=\Omega, \; \psi(T)=U, \; \psi_0(\Lambda)=\Omega, \; \psi_0(T_0)=U. \]
  Let $\xi=\psi^{-1}\circ\psi_0$. The situation is well explained by the following diagram.
  \begin{center}
  \begin{tikzpicture}[description/.style={fill=white,inner sep=2pt}]
  \matrix (m) [matrix of math nodes, row sep=2.5em,
  column sep=0.5em, text height=1.5ex, text depth=0.25ex, inner sep=6pt]
  { (\Lambda,T_0) & & (\Lambda,T) \\
    & (\Omega,U) & \\};
	  \path[->,font=\scriptsize]
		  (m-1-1) edge[|->] node[auto] {$\xi$} (m-1-3)
		  (m-1-1) edge[|->] node[below] {$ \psi_0\qquad $} (m-2-2)
		  (m-1-3) edge[|->] node[auto] {$ \psi $} (m-2-2);
  \end{tikzpicture}
  \end{center}
  So, $\xi(\Lambda)=\Lambda$ and $\xi(T_0)=T$. The affinity $\xi$ will be of the form $x\mapsto Bx+c$, for some $B\in\GL(d,\Z)$ and $c\in\Z^d$.
  Since $\xi(\Lambda)=\Lambda$, the submodule $\Z^k\subseteq \Z^d$ is mapped into itself by $\xi$. So $c$ belongs to $\Z^k$ (because it is the image of $0$, which belongs to $\Z^k$) and $B$ can be written as a block matrix in the following way:
  \[ B = \m{ B_1 & B_2 \\ 0 & B_3 }, \]
  where the block $B_1$ is $k\times k$ and the block $B_3$ is $(d-k)\times(d-k)$.
  Notice that, since $B$ is in $\GL(d,\Z)$, both $B_1$ and $B_3$ must have determinant equal to $1$ or $-1$.
  In particular, $B_1$ is in $\GL(k,\Z)$. Consequently, the affinity $\vartheta$ defined by $x\mapsto B_1x + c$ belongs to $\Aff(k,\Z)$.
  By construction, we also have that $\vartheta(\Lambda) = \Lambda$ and $\vartheta(T_0)=T$.
  Finally, we define
  \[ \chi(R) = \vartheta. \]
  
  As anticipated, we will now show that $\chi$ is injective.
  Suppose that we have $\chi(R_1)=\chi(R_2)=\vartheta$, for some $R_1,R_2\in\mathcal{S}$.
  Let $T_1$ and $T_2$ be the complete $\Lambda$-frames constructed throughout the execution of the function \Call{CFWF}{} applied to $(\Lambda,Q\cup R_1)$ and $(\Lambda,Q\cup R_2)$, respectively.
  As a consequence of what we proved above, $\vartheta(T_0)=T_1$ and $\vartheta(T_0)=T_2$. Thus, $T_1=T_2$.
  Assume now by contradiction that $R_1\neq R_2$. Since $R_1$ and $R_2$ are both subsets of $\Lambda$, and since $Q$ and $R_1$ (resp.\ $Q$ and $R_2$) are disjoint, we have that $(Q\cup R_1)\cap \Lambda \neq (Q\cup R_2)\cap \Lambda$.
  Then the values of $T_1$ and $T_2$, as they are initialized in line \ref{line:initialize-T} of Algorithm \ref{alg:canonical-form-with-frame}, are different. Therefore, $T_1$ and $T_2$ are different at the end of the execution too. This is a contradiction, because we proved that $T_1=T_2$.
  So $\chi$ is injective.
  
  Define now $\pi\colon G \to \GL(k,\Z)$ as the map that sends an affinity $\vartheta\in G$, of the form $x\mapsto Ex+c$, to its linear part $E\in\GL(k,\Z)$. The map $\pi$ is a group homomorphism, and $\ker\pi$ is the subgroup of $G$ consisting of the translations. Since $G$ is finite, it does not contain any non-trivial translation (because non-trivial translations have infinite order). This means that $\ker\pi$ is trivial, so $\pi$ is injective.
  
  By a classical result of Minkowski \cite{Minkowski1887}, for any fixed $d$ there exists a constant $h_d>0$ such that every finite subgroup of $\GL(d,\Z)$ has order $\leq h_d$ (see for instance \cite{serre2007bounds} for a proof in English).
  The group $G$ can be regarded (through the injective homomorphism $\pi$) as a subgroup of $\GL(d,\Z)$, so $G$ has order $\leq h_d$. Since we have built an injection $\chi\colon\mathcal{S}\to G$, the size of $\mathcal{S}$ is also $\leq h_d$. \qed
\end{proof}

\begin{lemma}
  \label{lemma:constant-frames}
  If the dimension $d$ is fixed, then the number of frames returned by any call to \Call{EquivariantFrames2}{} is $O(1)$. In other words, there exists a constant $h_d>0$ (depending on $d$) such that
  \[ | \Call{EquivariantFrames2}{\Lambda,Q} | \leq h_d \quad \forall\, \Lambda, Q. \]
\end{lemma}

\begin{proof}
  The constant $h_d$ will be the same as that of Lemma \ref{lemma:limited}.
  Let us analyze the possible return values of a call to \Call{EquivariantFrames2}{}.
  The size of the set returned in line \ref{line:first-case-end} is that of \Call{EquivariantFrames2}{$\Lambda',Q'$}; working by induction on the diameter of $\Lambda$, we can assume that such size is $\leq h_d$.
  The other possible return values are those assigned in lines \ref{line:F-first-assignment} and \ref{line:F-second-assignment}; in both cases, Lemma \ref{lemma:limited} assures that the size is $\leq h_d$. \qed
\end{proof}

\begin{theorem}
  \label{thm:complexity-canonical-frames2}
  If the dimension $d$ is fixed, then the asymptotic complexity of Algorithm \ref{alg:canonical-frames2} is $O(n \log^2 n \; s \, \mu(s))$, where $n$ and $s$ are as in the statement of Theorem \ref{thm:complexity-canonical-form-with-frame}.
\end{theorem}

\begin{proof}
  We show by induction that the execution of \Call{EquivariantFrames2}{$\Lambda,Q$} requires at most $\gamma f(w)\,n\log^2 n \,\log\delta\;\mu(s) + \beta$ operations, where $\gamma$ and $\beta$ are some constants (depending on $d$), $w = d-\dim\Span(Q)$, $\delta$ is the diameter of $\Lambda$ (which we will denote by $\diam(\Lambda)$), and the function $f$ is defined later in the proof (and depends on $d$). Once we prove this, we are done since $w \leq d = O(1)$ and $\log\delta = O(s)$.
  
  The induction is made on the triple $(|\Lambda|, \diam(\Lambda), |Q|)$. The ordering on such triples is the lexicographic one.
  
  {\bf Base case: $|\Lambda| = 1$.} Only lines \ref{line:begin}-\ref{line:return-base-case} are executed, and the total number of operations is $O(1)$. Thus, it is sufficient to choose the constant $\beta$ large enough.

  {\bf Inductive step.} Lines \ref{line:begin}-\ref{line:end-first-if} require $O(ns)$ operations.
  Let us now analyze lines \ref{line:congruence-condition2}-\ref{line:first-case-end}: their cost is $O(ns)$ plus the cost of the recursive call in line \ref{line:first-call2}, which is (by induction) $\gamma f(w)\,n\log^2 n \,(\log\delta - 1)\,\mu(s) + \beta$, since the diameter of $\Lambda$ is halved.
  
  Lines \ref{line:second-case-begin}-\ref{line:sort} require $O(ns)$ operations, since $h\leq 2^d = O(1)$.
  
  Now we turn to lines \ref{line:balanced-begin}-\ref{line:balanced-end}. Let $n_i = |\Lambda_i|$. The cost of the recursive calls alone (line \ref{line:second-call2}) is then, by induction,
  \begin{IEEEeqnarray*}{rCl}
    && \sum_{i=1}^h \big( \gamma  f(w)\, n_i \log^2 n_i\log\delta\;\mu(s) + \beta \big) \\
    &\leq & \gamma f(w)\, (n_1+\dots+n_h) \log^2 \frac{n}{2} \,\log \delta\;\mu(s) + h\beta \\
    &=& \gamma f(w)\, n \,(\log n -1)^2 \log\delta \; \mu(s) + O(1).
  \end{IEEEeqnarray*}
  By Lemma \ref{lemma:constant-frames}, the size of $\Lambda'$ in line \ref{line:big-union} is $O(1)$, so the number of operations required for lines \ref{line:F-first-assignment}-\ref{line:balanced-end} is $O(n\,\mu(s))$ (by Theorem \ref{thm:complexity-canonical-form-with-frame}) plus the cost of sets comparison.
  Comparing $O(1)$ sets of $O(n)$ elements requires $O(n\log n)$ comparisons of elements, each taking $O(s)$ operations, for a total of $O(ns\log n)$ operations.
  
  We finally turn to lines \ref{line:unbalanced-begin}-\ref{line:unbalanced-end}.
  The cost of the recursive calls in line \ref{line:third-call2} is
  \begin{IEEEeqnarray*}{rCl}
    && \sum_{i=1}^{h-1} \big( \gamma f(w)\, n_i \log^2 n_i \log\delta\;\mu(s) + \beta \big) \\
    &\leq & \gamma  f(w) \, (n_1+\dots+n_{h-1}) \log^2 \frac{n}{2} \,\log \delta\;\mu(s) + (h-1)\beta \\
    &\leq& \gamma f(w)\, \frac{n}{2}\, (\log n -1)^2\,\log\delta \;\mu(s) + O(1).
  \end{IEEEeqnarray*}
  By Lemma \ref{lemma:constant-frames}, the sets $\Lambda'$ and $E$ have size $O(1)$, so lines \ref{line:lambda'}-\ref{line:E} require $O(\mu(s))$ operations to be executed. As another consequence, there is a bound $\eta$ on the number of recursive calls in line \ref{line:fourth-call2}. Notice that $|Q\cup R| > |Q|$, because $\Lambda'$ is nonempty (since $h\geq 2$) and contains only points that don't belong to $\Span(Q)$ (thanks to line \ref{line:begin}). So, each of the calls of line \ref{line:fourth-call2} requires a number of operations bounded by
  \[ \gamma f(w-1)\, n\log^2 n \log\delta\;\mu(s) + \beta. \]
  Finally, $O(n\,\mu(s) + ns\log n)$ operations are required for lines \ref{line:F-second-assignment}-\ref{line:unbalanced-end} (as for lines \ref{line:F-first-assignment}-\ref{line:balanced-end}).
  
  We now define the function $f$ in the following way: $f(w) = (2\eta)^w$, where $\eta$ is the bound we introduced in the previous paragraph.
  
  Let us put everything together. We obtain the following results, depending on which lines are executed.
  \begin{itemize}
    \item Lines \ref{line:begin}-\ref{line:end-first-if} and \ref{line:congruence-condition2}-\ref{line:first-case-end} (all points of $\Lambda$ are congruent modulo $2$):
    \begin{IEEEeqnarray*}{rCl}
      && O(ns) + \gamma f(w)\,n\log^2 n \,(\log\delta - 1)\,\mu(s) + \beta \\
      &\leq & \alpha_1 n s + \gamma  f(w)\, n\log^2 n \log \delta\;\mu(s) - \gamma f(w)\,n\log^2 n\;\mu(s) \\
      &=& \gamma  f(w)\, n\log^2 n \log \delta\;\mu(s) + n\,\big(\alpha_1 s - \gamma  f(w) \log^2 n \;\mu(s)\big) \\
      &\leq& \gamma  f(w)\, n\log^2 n \log \delta\;\mu(s).
    \end{IEEEeqnarray*}
    The first inequality holds for any constant $\alpha_1$ such that $\alpha_1ns$ is greater than the term $O(ns)+\beta$.
    The last inequality holds for a sufficiently large value of $\gamma$, since $\mu(s)=\Omega(s)$.
    
    \item Lines \ref{line:begin}-\ref{line:end-first-if}, \ref{line:second-case-begin}-\ref{line:sort} and \ref{line:balanced-begin}-\ref{line:balanced-end} (balanced partition):
    \begin{IEEEeqnarray*}{rCl}
      && O(n\,\mu(s) + ns\log n) + \gamma f(w)\, n \,(\log n -1)^2 \log\delta \; \mu(s) \\
      &\leq& \alpha_2\, n\log n\;\mu(s) + \gamma f(w)\, n\log^2 n \log\delta\;\mu(s) - \gamma f(w)\,n\log n\,\log\delta\;\mu(s) \\
      &=& \gamma f(w)\, n\log^2 n \log\delta\;\mu(s) + n\log n\,\mu(s) \, \big( \alpha_2 - \gamma  f(w) \log\delta \big) \\
      &\leq& \gamma f(w)\, n\log^2 n \log\delta\;\mu(s).
    \end{IEEEeqnarray*}
    The first inequality holds for any constant $\alpha_2$ such that $\alpha_2\, n\log n\;\mu(s)$ is greater than the term $O(n\,\mu(s) + ns\log n)$.
    The last inequality holds for a sufficiently large value of $\gamma$.
    
    \item Lines \ref{line:begin}-\ref{line:end-first-if}, \ref{line:second-case-begin}-\ref{line:sort} and \ref{line:unbalanced-begin}-\ref{line:unbalanced-end} (unbalanced partition):
    \begin{IEEEeqnarray*}{rCl}
      && O(n\,\mu(s) + ns\log n) + \gamma f(w)\, \frac{n}{2}\, (\log n -1)^2\,\log\delta \;\mu(s) \,+\\
      && +\, \eta\, \gamma f(w-1)\, n\log^2 n \log\delta\;\mu(s) + \eta\beta \\
      &\leq& \alpha_3\, n\log n\,\mu(s) + \frac12 \, \gamma f(w)\, n\,\log n\, (\log n - 1) \log\delta \;\mu(s)\,+\\
      && +\, \eta\, \gamma f(w-1)\, n\log^2 n \log\delta\;\mu(s) \\
      &=& \alpha_3\, n\log n\,\mu(s) + \frac12 \, \gamma f(w)\, n\,(\log^2 n - \log n) \log\delta \;\mu(s) \,+\\
      && +\, \frac12\, \gamma f(w)\, n\log^2 n\log\delta\;\mu(s) \\
      &=& \alpha_3\, n\log n\,\mu(s) + \gamma f(w)\, n\,\log^2 n \log\delta\;\mu(s) - \frac12 \gamma f(w)\, n\log n \log\delta \;\mu(s)\\
      &=& \gamma f(w)\, n\log^2 n \log\delta\;\mu(s) + n\log n\,\mu(s) \left( \alpha_3 - \frac12 \gamma f(w)\, \log\delta \right) \\
      &\leq& \gamma f(w)\, n\log^2 n \log\delta \;\mu(s).
    \end{IEEEeqnarray*}
    The first inequality holds for any constant $\alpha_3$ such that $\alpha_3\, n\log n\,\mu(s)$ is greater than the term $O(n\,\mu(s) + ns\log n) + \eta\beta$.
    The second step follows from the identity $2\eta f(w-1) = f(w)$, which is an immediate consequence of the definition of $f$.
    The last inequality is true for a sufficiently large value of $\gamma$. \qed
  \end{itemize}
\end{proof}

\section{Canonical form for \texorpdfstring{$X_d$}{Xd}}
\label{sec:canonical-form}

Using the results of Sections \ref{sec:canonical-form-with-frame} and \ref{sec:canonical-frames}, we are now able to easily describe an algorithm to compute a canonical form for $X_d$.

\begin{algorithm}[Canonical form for $X_d$]
  \label{alg:canonical-form}
  This algorithm takes as input a finite set $\Lambda \subseteq \Z^d$, and returns a canonical form for $\Lambda$.
  
  \begin{alg}
    \Function{CanonicalForm}{$\Lambda$}
      \State $S \gets \Call{EquivariantFrames2}{\Lambda,\varnothing}$ \label{line:call-canonical-frames2}
      \ForAll{$R\in S$} \label{line:canonical-form-start-for}
	\State $(\Omega_R,U_R) \gets \Call{CanonicalFormWithFrame}{\Lambda,R}$
      \EndFor \label{line:canonical-form-end-for}
      \State \Return $\min \, \{ \, \Omega_R \mid R\in S \, \}$
    \EndFunction
  \end{alg}
\end{algorithm}

In words, Algorithm \ref{alg:canonical-form} first equivariantly computes a set $S$ of complete $\Lambda$-frames, using Algorithm \ref{alg:canonical-frames2}. Then, for each $\Lambda$-frame $R\in S$, it finds a corresponding canonical set $\Omega_R$. Finally, it returns the set which is lexicographically minimal among the computed ones.

\begin{theorem}
  \label{thm:correctness-canonical-form}
  The output of Algorithm \ref{alg:canonical-form} is a canonical form for $X_d$ with respect to the action of $\Aff(d,\Z)$.
  Moreover its worst-case asymptotic complexity is $O(n\log^2 n\, s\, \mu(s))$, where $n$ and $s$ are defined as in Section \ref{sec:preliminaries}.
\end{theorem}

\begin{proof}
  The first property of Definition \ref{def:canonical-form} is verified since Algorithm \ref{alg:canonical-form-with-frame} satisfies property 1 at the beginning of Section \ref{sec:canonical-form-with-frame}.
  The second property of Definition \ref{def:canonical-form} is an immediate consequence of Theorems \ref{thm:correctness-canonical-form-with-frame} and \ref{thm:correctness-canonical-frames2}.
  
  By Theorem \ref{thm:complexity-canonical-frames2}, the execution of line \ref{line:call-canonical-frames2} requires $O(n\log^2 n\, s\,\mu(s))$ operations.
  By Lemma \ref{lemma:constant-frames}, the size of $S$ is $O(1)$. Thus, by Theorem \ref{thm:complexity-canonical-form-with-frame}, the execution of lines \ref{line:canonical-form-start-for}-\ref{line:canonical-form-end-for} requires $O(n\,\mu(s))$ operations.
  The overall asymptotic complexity is therefore $O(n\log^2 n\, s\,\mu(s))$. \qed
\end{proof}

If Algorithm \ref{alg:canonical-form-with-frame} is modified so that it also returns the affinity $\psi$, then Algorithm \ref{alg:canonical-form} can be also modified to return an affinity which maps $\Lambda$ to its canonical form.
In this way, if two sets have the same canonical form, it is possible to explicitly construct an affinity which maps one to the other.

\section{Canonical form of Alexander polynomials}
\label{sec:alexander}

We now turn to the application of our algorithm to the computation of a canonical form of the Alexander polynomial of a group.
Let us first recall the construction of such polynomial, as given by Fox \cite{fox1953free,fox1954calculus2}.

Let $G = \langle x_1, \dots, x_n \mid r_1, \dots, r_k \rangle $ be a finitely presented group. Using ``free differential calculus'' \cite{fox1953free}, in \cite{fox1954calculus2} Fox defines the \emph{Jacobian} $J$ of the presentation, which is a $k\times n$ matrix with entries in the group ring $\Z G$.
Assume from now on that the abelianized group $H = G/G'$ is isomorphic to $\Z^d$, via some fixed isomorphism $\psi\colon H \to \Z^d$. Let $\phi\colon G \to H$ be the abelianization map.
Then the Jacobian $J$ can be mapped to the \emph{Alexander matrix} $\mathcal{A} = \psi\circ\phi(J)$ with entries in the group ring of $\Z^d$, i.e.\ in the Laurent polynomial ring $R = \Z[t_1^{\pm 1}, \dots, t_d^{\pm 1}]$.
As pointed out by Fox, such construction generalizes Alexander's classical construction \cite{alexander1928topological}.
The ideal $\epsilon_i$ generated by the minor determinants of $\mathcal{A}$ of order $n-i$ is called the \emph{$i$-th elementary ideal} of $\mathcal{A}$.
The elementary ideals are independent of the chosen presentation, but they do depend on $\psi$.

Assume now that $G$ admits a presentation with more generators than relations.
Important examples are fundamental groups of the complement of links \cite{rolfsen1976knots}.
Then for $d=1$ the elementary ideal $\epsilon_1$ is principal, and for $d=2$ it is the product of a principal ideal and the \emph{fundamental ideal} $(t_1-1, \dots, t_d-1)$.
In both cases, a generator $\Delta(t_1,\dots,t_d)$ of the obtained principal ideal is called an \emph{Alexander polynomial}.
Such polynomial is defined uniquely up to multiplication by monomials of the form $\pm t_1^{\lambda_1}\cdots t_d^{\lambda_d}$, if the isomorphism $\psi$ is fixed.
As show in \cite{bellettini2015shape}, changing $\psi$ by some linear automorphism $A$ of $\Z^d$ does not affect the fundamental ideal, but affects $\epsilon_1$ (and therefore the Alexander polynomial $\Delta$) transforming the exponents vector of every monomial by $A$.
Therefore an invariant of $G$ is given by a canonical form of the Alexander polynomial with respect to to the action of $\Aff(d, \Z)$ and change of sign.

We are now going to illustrate how to adjust Algorithm \ref{alg:canonical-form} in order to compute a canonical form of any Laurent polynomial $\Delta(t_1,\dots,t_d) \in R$.
The possibility of changing the sign of the entire polynomial can be easily settled (e.g.\ choosing the sign so that the leading term has positive coefficient), so we are going to focus on the action of $\Aff(d,\Z)$ only.
Let us write
\[ \Delta(t_1,\dots,t_d) = \sum_{m_1,\dots,m_d\in \Z} \alpha_{m_1,\dots,m_d} \, t_1^{m_1}\cdots t_d^{m_d}, \]
where only a finite number of coefficients $\alpha_{m_1,\dots,m_d}$ is non-zero.
The action of $\Aff(d,\Z)$ on $R$ is by means of $\Z$-linear automorphisms, so it can be described on a single monomial $t_1^{m_1}\cdots t_d^{m_d}$. An integer affinity $\varphi\in\Aff(d,\Z)$ maps the monomial $t_1^{m_1}\cdots t_d^{m_d}$ to the monomial $t_1^{p_1}\cdots t_d^{p_d}$, where
$(p_1,\ldots,p_d)^t=\varphi(m_1,\ldots,m_d)^t$.
% \[ \m{p_1 \\ \vdots \\ p_d} = \varphi \m{m_1 \\ \vdots \\ m_d}. \]
A polynomial $\Delta(t_1,\dots,t_d)$ can be viewed as a finite set of points in $\Z^d$ (the set of vectors $(m_1,\dots,m_d)^t$ for which the coefficient $\alpha_{m_1,\dots,m_d}$ is non-zero) with an integer coefficient $\alpha_{m_1,\dots,m_d}$ associated to each point.
Under this identification, the action of $\Aff(d,\Z)$ is precisely the natural action on the subsets of $\Z^d$, with the subtlety that these subsets are \emph{weighted}. In analogy with the previous notations, let $X_d^w$ be the set of finite weighted subsets of $\Z^d$ (weights are given by non-zero integers).

Algorithms \ref{alg:canonical-form-with-frame}, \ref{alg:canonical-frames2} and \ref{alg:canonical-form}, without any change, work as well for the case of $X_d^w$.
What changes is that every operation involving elements of $X_d$ must now involve elements of $X_d^w$.
For instance \Call{CanonicalFormWithFrame}{} must output, as the first element of the pair, a weighted set.
Therefore, in lines \ref{line:F-second-assignment}-\ref{line:unbalanced-end} of Algorithm \ref{alg:canonical-frames2}, the minimality check must take into account the weights as well, so that two sets with the same elements but with different weight are not treated as equal.

For simplicity we assume that the size of the binary representation of each weight is also $O(s)$, so that the comparison of weights takes at most as much as the comparison of the corresponding points.
Then the cost of dealing with weights is fully covered by the original complexity bound.
We omit the proof of correctness and the proof of the complexity bound since they are entirely analogous to the case of $X_d$, which was thoroughly described in the previous sections.

Apart from the intrinsic interest of computing a canonical Alexander polynomial, what we have done might be also applied to the more general problem of testing isomorphism between finitely presented groups.
This is a classical problem in computational group theory \cite{eick2015computational,holt1992testing,holt2005handbook,sims1994computation}, and is not decidable in general (this follows from a work of Novikov on the unsolvability of the word problem \cite{novikov1955algorithmic}).
To try to show that two finitely presented group are not isomorphic, the usual approach consists in analyzing their finite quotients of small order \cite{holt1992testing}: if these are not the same, then the two groups cannot be isomorphic.
Computing a canonical Alexander polynomial could be a quite different method to try to distinguish non-isomorphic groups.
We did not investigate the relation between this approach and the usual ones. In particular, we do not know if there exist non-isomorphic groups that can be distinguished by their canonical Alexander polynomial but not (easily) by their small finite quotients.

In the case of fundamental groups of the complement of links, the parameter $d$ equals the number of components of the link.
In particular, large classes of interesting and well-studied groups arise even for very small values of $d$.
On the other hand, the size of the Alexander polynomial can grow fast in terms of the number of crossings.
This justifies the choice to fix the value of $d$ in the analysis of the asymptotic complexity.

\section{Conclusions}
\label{sec:conclusion}

The algorithm we have presented in Sections \ref{sec:canonical-form-with-frame}-\ref{sec:canonical-form} computes a canonical form of subsets of $\Z^d$ up to affinity, and has asymptotic complexity $O(n\log^2 n\, s\,\mu(s))$ for any fixed dimension $d$.
In particular, the problem we consider is fixed-parameter tractable with respect to the dimension $d$.
The dependence on $n$ is obviously optimal up to logarithmic factors. 

We chose to omit explicit analysis on how the multiplicative constant grows in terms of $d$ (the dependence is probably at least doubly exponential).
There are many possible improvements that can lower such constant. For instance, in Algorithm \ref{alg:canonical-frames2}, one can further exploit the canonical partition $\{\Lambda_1,\dots,\Lambda_h\}$ to significantly reduce the number of frames $R$ considered in lines \ref{line:F-first-assignment}-\ref{line:balanced-end} and \ref{line:F-second-assignment}-\ref{line:unbalanced-end}.
However, since such improvements affect only the constant, we have preferred to ignore them in order to keep the pseudocode as essential as possible.

It is finally worth noticing that the presented algorithms can be easily modified to also output an affinity which sends the input set $\Lambda$ to its canonical form. Therefore, when deciding whether two given sets $\Lambda$ and $\Lambda'$ are in the same orbit with respect to the action of $\Aff(d,\Z)$, in case of an affirmative answer it is possible to explicitly obtain an affinity that sends $\Lambda$ to $\Lambda'$.

\section*{acknowledgements}
I would like to thank my father Maurizio Paolini for having introduced me to the problem and for his valuable suggestions.
I would also like to thank Luca Ghidelli for the useful discussions and for having carefully read this paper.
Then I would like to thank professors Patrizia Gianni, Carlo Traverso and Giovanni Gaiffi for their advice.
Finally I thank the referees, for the thorough revisions and for pointing out interesting references and connections.

\bibliography{bibliography}{}
\bibliographystyle{spmpsci}

\end{document}